\UseRawInputEncoding
\documentclass{article}
\pdfoutput=1
\usepackage{hyperref}
\usepackage{pdfpages}
\usepackage{array, longtable}
\usepackage{xcolor}
\usepackage{soul}
\usepackage{amsmath,amsthm,amssymb}
\usepackage{array, longtable}
\usepackage{xcolor}
\usepackage{soul}
\usepackage{booktabs}
\usepackage{makecell}
\usepackage{diagbox}
\usepackage{graphicx}
\usepackage{caption}
\usepackage{float}

\newtheorem{Theorem}{Theorem}[section]
\newtheorem{lem}[Theorem]{Lemma}

\newtheorem{Definition}[Theorem]{Definition}
\newtheorem{Corollary}[Theorem]{Corollary}

\newtheorem{Construction}[Theorem]{Construction}
\newtheorem{Example}[Theorem]{Example}
\numberwithin{equation}{section}
\begin{document}
\title{{On non-expandable cross-bifix-free codes\LARGE }}

\author{Chunyan Qin$^1$,    Bocong Chen$^1$ and Gaojun Luo$^2$\footnote{E-mail addresses: bocongchen@foxmail.com (B. Chen),
chunyan\_qin@126.com (C. Qin) and  gaojun.luo@ntu.edu.sg (G. Luo).}}  


\date{\small
$1.$ School of Mathematics, South China University of Technology, Guangzhou, 510641, China\\
$2.$  School
of Physical and Mathematical Sciences, Nanyang Technological University,
639798, Singapore\\
}


\maketitle

\begin{abstract}
A  cross-bifix-free code of length $n$ over $\mathbb{Z}_q$  is defined as a non-empty subset of $\mathbb{Z}_q^n$
satisfying  that the prefix  set of each codeword is disjoint from the suffix set of every codeword.
Cross-bifix-free codes have found important applications in digital  communication systems.
One of the main research problems on cross-bifix-free codes is to construct cross-bifix-free codes as large as possible in size.
Recently, Wang and Wang  introduced a  family of cross-bifix-free codes $S_{I,J}^{(k)}(n)$, which is a generalization of
the classical cross-bifix-free codes studied early  by Lvenshtein, Gilbert and Chee {\it et al.}.
It is known that $S_{I,J}^{(k)}(n)$ is nearly optimal in size and
$S_{I,J}^{(k)}(n)$  is non-expandable if $k=n-1$ or $1\leq k<n/2$.
In this paper, we first show that  $S_{I,J}^{(k)}(n)$ is non-expandable if and only if   $k=n-1$ or $1\leq k<n/2$,
thereby improving the results in [Chee {\it et al.}, IEEE-TIT, 2013] and [Wang and Wang,  IEEE-TIT, 2022].
We then
construct a new family of cross-bifix-free codes $U^{(t)}_{I,J}(n)$ to expand $S_{I,J}^{(k)}(n)$ such that
the resulting larger code $S_{I,J}^{(k)}(n)\bigcup U^{(t)}_{I,J}(n)$ is a non-expandable
cross-bifix-free code whenever $S_{I,J}^{(k)}(n)$ is expandable.
Finally, we present an explicit  formula for the size of
$S_{I,J}^{(k)}(n)\bigcup U^{(t)}_{I,J}(n)$.



\end{abstract}

\section{Introduction}
The class of cross-bifix-free codes (also named non-overlapping codes or  strongly regular codes in  literature)
arises in dealing with the problem of establishing synchronization between the transmitter
and  receiver in  digital communication systems.
The historical engineering method began with the introduction of {\it bifix}, a name proposed by Nielsen in \cite{Nielsen}.
In order to achieve fast and reliable synchronization in binary data streams, Massey \cite{Massey} introduced
the notion of {\it bifix-free} synchronization word.
A bifix-free word is a sequence of symbols satisfying that no prefix of any length of the word is identical with any suffix
of the word.
Bajic {\it et al.} \cite{Bajic1} showed that the distributed sequence entails a simultaneous search for
a particular set of synchronization words: A set of  words not only satisfies
that each word is bifix-free but also satisfies that no prefix of any length of any word in the set is a suffix of any other word in the set.
The particular set of synchronization words is termed as a {\it cross-bifix-free code}.

The research   of the frame synchronization  involving cross-bifix-free codes is to establish and maintain the connection between the transmitter and receiver of digital communication systems \cite{Bajic1,A.J,Shen,Stefanovic}.
Cross-bifix-free codes are useful in the sense that the codewords with an error or in a state of a certain decoding automaton do not propagate to subsequent incorrect decoding codewords. Cross-bifix-free codes have also found applications in DNA storage systems \cite{Levy,Yazdi}, pattern matching \cite{Crochemore} and automata theory \cite{Berstel}.

One of the main research problems on cross-bifix-free codes is to construct cross-bifix-free codes as large as possible in size.
The first systematic construction of cross-bifix-free codes $S_{q}^{(k)}(n)$ was due to  Levenshtein \cite{V.I.1,V.I.2}
(see Construction \ref{S1} in Section $2$) in 1964, which was also studied  by Gilber \cite{Gilbert} and rediscovered by Chee {\it et al}. \cite{Chee}.
Bajic  {\it et al.}  \cite{Bajic2,Bajic3} presented a general method to produce cross-bifix-free codes.
Bilotta {\it et al.} \cite{Bilotta} provided a construction of  binary cross-bifix-free codes based on Dyck paths.
It was shown in \cite{Chee} that $S^{(k)}_{q}(n)$ is nearly optimal in size and  a sufficient condition on $k$ was given to ensure  that $S^{(k)}_{q}(n)$ is non-expandable.
Blackburn \cite{Blackburn15} provided a simple construction for a class of cross-bifix-free codes with optimal cardinality.
Barucci {\it et al.} \cite{Barcucci} proposed a construction of $q$-ary $(q\geq3)$ cross-bifix-free codes based on colored Motzkin paths.
Anselmo,  Giammarresi and Madonia \cite{Anselmo} established a general constructive method for non-expandable cross-bifix-free codes by using some structural properties of non-expandable cross-bifix-free codes.
Recently, Wang and Wang \cite{Wang} introduced a  family of $q$-ary cross-bifix-free codes $S_{I,J}^{(k)}(n)$, which is a generalization of
the classical cross-bifix-free codes $S^{(k)}_{q}(n)$.
A number of interesting results were obtained in
\cite{Wang}; for example, \cite{Wang} gave an exact formula for the number of fixed-length words that do not contain the codewords in a
variable-length
cross-bifix-free code as subwords, which solves an open problem by Bilotta \cite{Bilotta0}.
Blackburn {et al.} \cite{Blackburn23} provided some general bounds and  several constructions for certain binary codes
in which over-laps of certain specified sizes are forbidden.
Very recently, Stanovnik,   Mo\v{s}kon and   Mraz \cite{Stanovnik} characterized the maximal cross-bifix-free codes,
formulated  the maximum cross-bifix-free code problem as an integer optimization problem
and determined necessary conditions for optimality of a cross-bifix-free code.

In this paper, we only consider  cross-bifix-free codes with fixed length. Readers who are interest in variable-length cross-bifix-free codes
can refer to \cite{Bilotta0,Wang}. Motivated by the prior works, especially the works of \cite{Chee} and \cite{Wang}, the objective of this paper is to further explore the non-expandable property of the $q$-ary code $S_{I,J}^{(k)}(n)$.
Our starting point is \cite[Theorem 1]{Wang} which says that $S_{I,J}^{(k)}(n)$  is non-expandable if $k=n-1$ or $1\leq k<n/2$.
After carefully analyzing  the combinatorial structure of $S_{I,J}^{(k)}(n)$,
we  show that  $S_{I,J}^{(k)}(n)$ is non-expandable if and only if   $k=n-1$ or $1\leq k<n/2$,
which improves  the results in \cite[Theorem 3.1]{Chee} and \cite[Theorem 1]{Wang}.
In other words, $S_{I,J}^{(k)}(n)$ is  expandable precisely when   $n/2\leq k\leq n-2.$
We then
construct a new family of cross-bifix-free codes $U^{(t)}_{I,J}(n)$ to expand $S_{I,J}^{(k)}(n)$. The new
 code $S_{I,J}^{(k)}(n)\bigcup U^{(t)}_{I,J}(n)$ is a non-expandable
cross-bifix-free code whenever $S_{I,J}^{(k)}(n)$ is expandable.
Finally, we present an enumerative formula for the size of
$S_{I,J}^{(k)}(n)\bigcup U^{(t)}_{I,J}(n)$ when $n/2\leq k\leq n-2.$
We remark that $S_{I,J}^{(k)}(n)\bigcup U_{I,J}^{(t)}(n)$
is a disjoint union and that the size of the $q$-ary code $S_{I,J}^{(k)}(n)$ is equal to $q^{n-k-2}|I|^k|J|^2$ since $k\geq n/2$.
Consequently, the problem of counting the size of $S_{I,J}^{(k)}(n)\bigcup U_{I,J}^{(t)}(n)$  is fully converted to that of
counting the size of $U_{I,J}^{(t)}(n)$.
We find explicit formulas for the size of $U_{I,J}^{(t)}(n)$ in the case where $n/2\leq k<3n/4-1/2$.
For the rest cases, recurrence relations
on the size of $U_{I,J}^{(t)}(n)$ are presented.
As an illustrated example, we list the values of $|S_2^{(k)}(n)|$ and $|S_2^{(k)}(n)\bigcup U_{2}^{(t)}(n)|$
in Tables $1$ and $2$ when $5\leq n\leq 17$ and $n/2\leq k\leq n-2$.
It is clear from Tables $1$ and $2$ that the gap between $|S_2^{(k)}(n)\cup U_2^{(t)}(n)|$ and  $|S_2^{(k)}(n)|$ is sometimes rather large, such as  $|S_2^{(15)}(17)|=1$ and
$|S_2^{(15)}(17)\bigcup U_{2}^{(2)}(17)|=1433$.

This paper is organized as follows. In Section 2, we review some basic definitions and notations related to the notions of   bifix-free words and cross-bifix-free codes. In Section 3, we first analyze the non-expandable property of the code
$S_{I,J}^{(k)}(n)$. Then we give a systematic way to produce a cross-bifix-free code $U_{I,J}^{(t)}(n)$ and we finally show that
$S_{I,J}^{(k)}(n)\bigcup U^{(t)}_{I,J}(n)$ is a non-expandable
cross-bifix-free code whenever $S_{I,J}^{(k)}(n)$ is expandable.
In Section 4, we derive an enumerative formula for the size of
$S_{I,J}^{(k)}(n)\bigcup U^{(t)}_{I,J}(n)$.
We conclude this paper with remarks and some possible future works in Section 5.

\section{Preliminaries}
In this section, we  review the notion of  cross-bifix-free codes and  represent the classical family of  cross-bifix-free codes
$S^{(k)}_{q}(n)$ as mentioned in the first section.
We first fix the notation that will be used throughout the paper.
As usual, let $\mathbb{Z}_{q}=\{0,1,2,\cdots,q-1\}$ where $q>1$ is a positive integer.
The cardinality (equivalently the size) of a finite set $S$ is denoted by $|S|$ and
the empty set is denoted by $\emptyset$.
A consecutive sequence of $m$ elements $b\in\mathbb{Z}_q$ is denoted by the short form $b^m$. For example, the vector (or called word)
$(1,1,0,0,0,1,0,0)$ can be expressed as a short form $(1^2,0^3,1,0^2)$. If $m=0$, then $b^m$ is used to denote the absence of  the element $b$.
For a vector (or string) $\alpha$, we also use the notation $|\alpha|$ to denote the length of the vector.

Let $\mathbb{Z}_q^n$ be the set of all $n$-tuples whose coordinates belong to $\mathbb{Z}_q$.
A $q$-ary {\em code} $C$ of length $n$  is simply defined as a non-empty subset of $\mathbb{Z}_q^n$.
At this time, the vectors belonging to $C$ are also called {\em codewords}.
The vector $(x_1,x_2,\cdots,x_r)\in\mathbb{Z}_{q}^{r}$ is called {\em $C$-free} if $r<n$, or if $r\geq n$, $(x_i,x_{i+1},\cdots,x_{i+n-1})\notin C$ holds for all $i=1,2,\cdots,r-n+1$
(in other words, there is no $n$-consecutive coordinates in $(x_1,x_2,\cdots,x_r)$ which forms an $n$-tuple belonging to $C$). For example, taking $q=2,r=5$ and $C=\{(0,0,0,0)\}$, the vector $(1,0,0,0,1)$ is $0^4$-free , but it is not $0^3$-free.

Before defining the notion of  cross-bifix-free codes, we need the following definitions.
\begin{Definition}
Let $n,q$ be integers strictly greater than $1$. For a vector $\mathbf{u}\in\mathbb{Z}_{q}^n$, the set of prefixes of $\mathbf{u}$ is
defined as ${\rm Pre}(\mathbf{u})=\{(u_1,u_2,\cdots,u_i)\,|\,1\leq i\leq n-1\}$ and the set of suffixes of $\mathbf{u}$ is defined as ${\rm Suf}(\mathbf{u})=\{(u_i,u_{i+1},\cdots,u_n)\,|\,2\leq i\leq n\}.$
\end{Definition}
For example,   the set of prefixes of $(00101)$ is ${\rm Pre}(00101)=\{0,00,001,0010\}$ and the set of suffixes of $(00101)$ is ${\rm Suf}(00101)=\{1,01,101,0101\}$.

\begin{Definition}\label{bifix-free}
A vector $\mathbf{u}\in\mathbb{Z}_{q}^n$ is called bifix-free if ${\rm Pre}(\mathbf{u})\bigcap {\rm Suf}(\mathbf{u})=\emptyset$.
Namely, a vector
$\mathbf{u} \in \mathbb{Z}_{q}^{n}$ is bifix-free if it cannot be factorized as $\mathbf{u}=\alpha\omega \alpha$, where $\alpha$ is necessarily a non-empty string, whereas $\omega$ can  be empty.
\end{Definition}
For example, the binary vectors $(0001001)$ and $(0001101)$ are both bifix-free; however, $(0010001)$ is not bifix-free since $001\in \mbox{Pre$(0010001)$}\bigcap \mbox{Suf$(0010001)$}$.
We have ready to present the definition of a cross-bifix-free code.
\begin{Definition}
A code $C\subseteq \mathbb{Z}_{q}^n$ is said to be cross-bifix-free   (or non-overlapping) if the code $C$ satisfies
${\rm Pre}(\mathbf{u})\bigcap {\rm Suf}(\mathbf{v})=\emptyset$ for all $\mathbf{u},\mathbf{v}\in C$ ($\mathbf{u}$ and $\mathbf{v}$ can be identical).
\end{Definition}
For example, the code $\{0001001,0001101\}$ is  cross-bifix-free,
while the code $\{0001001,0010001\}$ is not  cross-bifix-free  since $0001\in \mbox{Pre$(0001001)$}\bigcap \mbox{Suf$(0010001)$}$.
By  the very definition of a cross-bifix-free code, the prerequisite for $C\subseteq \mathbb{Z}_{q}^n$ to be a cross-bifix-free code is that each codeword in $C$ is bifix-free.

We now restate the classical  construction of the cross-bifix-free code $S_{q}^{(k)}(n)$, which was studied extensively in
\cite{Chee,V.I.1,V.I.2,Gilbert},

\begin{Construction}\label{S1}
Let $n>1,q>1$ be  integers  and $1\leq k\leq n-1$. Denote by $S_{q}^{(k)}(n)$ the set of all vectors $(s_1,s_2,\cdots,s_n)\in\mathbb{Z}_{q}^n$ that satisfies the following two conditions:
\begin{itemize}
\item[$(1)$] $s_1=s_2=\cdots=s_k=0$, $s_{k+1}\neq0$  and $s_n\neq0$;

\item[$(2)$] the subsequence $(s_{k+2},s_{k+3},\cdots,s_{n-1})$ is $0^k$-free.
\end{itemize}
Then $S_{q}^{(k)}(n)$ is a cross-bifix-free code.
\end{Construction}

Recently, Wang and Wang \cite{Wang} generalized $S_{q}^{(k)}(n)$ to $S_{I,J}^{(k)}(n)$.
More precisely, the sets $I$ and $J$ are called a {\em bipartition} of $\mathbb{Z}_{q}$ if $I\bigcup J=\mathbb{Z}_{q}$ and $I\bigcap J=\emptyset$, where $I, J$ are non-empty sets.

\begin{Construction}\label{S2}
Let $n>1,q>1$ be integers  and   $1\leq k\leq n-1$. Let  $I,J$ form a bipartition of $\mathbb{Z}_{q}$. Denote by $S_{I,J}^{(k)}(n)$ the set of all vectors
$(s_1,s_2,\cdots,s_n)\in\mathbb{Z}_{q}^n$ that satisfies the following two conditions:

\begin{itemize}
\item[$(1)$] $(s_1,s_2,\cdots,s_k)\in I^k$,   $s_{k+1}\in J $ and $s_n\in J$;

\item[$(2)$] the subsequence $(s_{k+2},s_{k+3},\cdots,s_{n-1})$ is $I^k$-free.

\end{itemize}
Then $S_{I,J}^{(k)}(n)$ is a cross-bifix-free code.

\end{Construction}

Clearly,
taking $I=\{0\}$ and $J=\{1,\cdots,q-1\}=[q-1]$ in Construction \ref{S2}, one has $S_{\{0\},{[q-1]}}^{(k)}(n)=S_{q}^{(k)}(n)$,
and thus Construction \ref{S1} can be viewed as a special case of Construction \ref{S2}.
A $q$-ary cross-bifix-free  code $C\subseteq \mathbb{Z}_q^{n}$ is said to be {\em non-expandable} if $C\bigcup\{\mathbf{x}\}$ is
no longer a cross-bifix-free code for each $\mathbf{x}\in \mathbb{Z}_q^{n}\setminus C$. Otherwise, it is said to be {\em expandable}.
Chee et al. \cite[Theorem 3.1]{Chee} showed that   $S_{q}^{(k)}(n)$ is
a non-expandable cross-bifix-free code for certain suitable values of  $k$.
Wang and Wang corrected   \cite[Theorem 3.1]{Chee}, and
more generally, they have the following result \cite[Theorem 1]{Wang}.
\begin{lem}\label{non-expandable1}
Let notation be the same as before. The $q$-ary cross-bifix-free code $S_{I,J}^{(k)}(n)$
given as in Construction \ref{S2} is non-expandable if $k=n-1$ or $1\leq k<n/2$.
\end{lem}
The staring point of this paper is based on Lemma \ref{non-expandable1}.
We aim to know whether or not the conditions on $k$ listed in Lemma  \ref{non-expandable1} are necessary and sufficient; if yes,
can one find a systematic way to expand $S_{I,J}^{(k)}(n)$ to   a non-expandable one in the case where $n/2\leq k\leq n-2$?
In the next section, we give  affirmative answers to these questions.

\section{Construct new  cross-bifix-free codes and expand  $S_{I,J}^{(k)}(n)$}
The objective of this section is twofold. First, we show that the conditions on $k$ listed in Lemma  \ref{non-expandable1} are
not only sufficient but also necessary. Second, we construct a new family of cross-bifix-free codes
\big(denoted by $U_{I,J}^{(t)}(n)$\big) such that
$S_{I,J}^{(k)}(n)\bigcup U_{I,J}^{(t)}(n)$ is a non-expandable cross-bifix-free code whenever $S_{I,J}^{(k)}(n)$ is expandable.

We first establish the following lemma, which plays an important role in achieving our first goal.
Note, in the case $n/2\leq k\leq n-2$,  that it is reasonable for us to restrict  $n\geq 5$
(we do not consider the senseless cases $n=1,2,3$; if $n=4$, then $k=2$, and one can easily check that $S_{I,J}^{(2)}(4)$ is non-expandable
for any $q>1$ and $I,J$ being a bipartition of $\mathbb{Z}_q$). For convenience, following the terminology in \cite{Wang}, we say that
$\textbf{x}=(x_1,\cdots,x_n)\in \mathbb{Z}_q^n$ overlaps with $S_{I,J}^{(k)}(n)$ if
$S_{I,J}^{(k)}(n)\bigcup\{\textbf{x}\}$ is no longer a cross-bifix-free code.

\begin{lem}\label{theorem1}
Let notation be the same as before.
The $q$-ary cross-bifix-free code $S_{I,J}^{(k)}(n)$ given by Construction \ref{S2} is expandable if $n\geq 5$ and $n/2\leq k\leq n-2$.
\end{lem}
\begin{proof}
Our ultimate goal is to find  a suitable vector $\textbf{x}=(x_1,\cdots,x_n)\in \mathbb{Z}_q^n\setminus S_{I,J}^{(k)}(n)$
such that $S_{I,J}^{(k)}(n)\bigcup\{\textbf{x}\}$ is still a cross-bifix-free code in the case where $n\geq 5$ and $n/2\leq k\leq n-2$.
By Construction \ref{S2}, if a $q$-ary vector $\textbf{x}=(x_1,\cdots,x_n)$
begins with $x_1\in{J}$ then $\textbf{x}$ overlaps with   $S_{I,J}^{(k)}(n)$, since the $1$-length prefix of  $\textbf{x}$ is also a
$1$-length suffix of some codeword $\textbf{y} \in{S_{I,J}^{(k)}(n)}$.
Similarly, we cannot expand any vector ending with $x_{n}\in{I}$ to $S_{I,J}^{(k)}(n)$.
Therefore, we only need to consider $\textbf{x}\in ({I\times\mathbb{Z}_{q}^{n-2}\times J})\setminus{S_{I,J}^{(k)} (n)}$. We consider the following two cases separately,
 according to $n/2\leq k<n-2$ or $k=n-2$.

{\bf Case 1)} $n/2\leq k< n-2$.
Note that if $n=5$, then there is no integer $k$ satisfying $n/2\leq k< n-2$. Thus, in this case $n\geq6$ by our assumption $n\geq5$.
Take a specific  vector $\textbf{x}\in I^{\ell}\times J^2 \times \mathbb{Z}_{q}^{n-k-3} \times I \times J^{k-\ell}$, where
$\ell$ is a positive integer satisfying $n-k-1 \leq \ell<k<n-2$.
Clearly, we have $\textbf{x}=(x_1,\cdots,x_n)\in \mathbb{Z}_q^n\setminus S_{I,J}^{(k)}(n)$.
First we claim that $\textbf{x}$ is
bifix-free (see Definition \ref{bifix-free}), and then we assert that   $\textbf{x}$ can be expanded to $S_{I,J}^{(k)}(n)$
preserving that $S_{I,J}^{(k)}(n)\bigcup\{\textbf{x}\}$ is still a cross-bifix-free code.

The first claim is easy to see:
Since $n-k-1 \leq \ell<k<n-2$,
we have $I^{\ell}\times J^2 \times \mathbb{Z}_{q}^{n-k-3} \times I \times J^{k-\ell}\subseteq S_{I,J}^{(\ell)}(n)$.
By Construction \ref{S2}, $S_{I,J}^{(\ell)}(n)$ is a cross-bifix-free code. It follows that  $I^{\ell}\times J^2 \times \mathbb{Z}_{q}^{n-k-3} \times I \times J^{k-\ell}$ is cross-bifix-free. Therefore, $\textbf{x}\in I^{\ell}\times J^2 \times \mathbb{Z}_{q}^{n-k-3} \times I \times J^{k-\ell}$
is  bifix-free, as claimed.

We continue to show that  $S_{I,J}^{(k)}(n)\bigcup\{\textbf{x}\}$ is  a cross-bifix-free code. To this end,
let the vector $\textbf{x}\in I^{\ell}\times J^2 \times \mathbb{Z}_{q}^{n-k-3} \times I \times J^{k-\ell}$ $(n-k-1 \leq \ell<k <n-2)$
be written as $\textbf{x}=\alpha\beta$, where $\alpha$ and $\beta$ are non-empty strings.
Take a typical element  $\textbf{v}\in S_{I,J}^{(k)}(n)$ and read $\textbf{v}$ as $\textbf{v}=\alpha{'}\beta{'}$,
where $\alpha{'}$ and $\beta{'}$ are non-empty strings with $|\alpha|=|\beta{'}|$.

If $|\alpha|=|\beta{'}|=j\leq \ell$, then $\alpha\in I^{j}$ and
the last coordinate of $\beta{'}$ belongs to $J$, giving $\alpha\neq \beta{'}$.

If $\ell<|\alpha|=|\beta{'}|=j<n-k+\ell$,
then
$\alpha\in I^{\ell}\times J~(j=\ell+1)$ or $\alpha\in I^{\ell}\times J^2 \times \mathbb{Z}_{q}^{j-\ell-2}~(\ell+2\leq j<n-k+\ell)$
and
$\beta{'}\in I^{j-(n-k)} \times J \times \mathbb{Z}_{q}^{n-k-2} \times J$ (here we use the assumption $j>\ell\geq n-k-1$).
As $j-(n-k)<\ell$, we get $\alpha\neq \beta{'}$.

If $|\alpha|=|\beta{'}|=j=n-k+\ell$, then $\alpha\in I^{\ell}\times J^2\times \mathbb{Z}_{q}^{n-k-3}\times I$
and $\beta{'}\in I^{\ell} \times J \times \mathbb{Z}_{q}^{n-k-2} \times J$, leading to
$\alpha\neq \beta{'}$.

If $n-k+\ell<|\alpha|=|\beta{'}|=j\leq n-1$, then $\alpha\in I^{\ell}\times J^2 \times \mathbb{Z}_{q}^{n-k-3}\times I\times J^{j-\ell-(n-k)}$
and
$\beta{'}\in I^{j-(n-k)} \times J \times \mathbb{Z}_{q}^{n-k-2} \times J$. Since $j-(n-k)>\ell$, we have $\alpha\neq \beta{'}$.

Therefore, we see that $\rm{Pre}(\textbf{x})\cup \rm{Suf}(\textbf{v})=\emptyset$.
The reasoning for $\rm{Pre}(\textbf{v})\cup \rm{Suf}(\textbf{x})=\emptyset$ where $\textbf{v}\in S_{I,J}^{(k)}(n)$
with $n/2\leq k< n-2$ is quite analogous to the just considered case $\rm{Pre}(\textbf{x})\cup \rm{Suf}(\textbf{v})=\emptyset$, by considering $\textbf{x}=\alpha\beta$ and $\textbf{v}=\alpha{'}\beta{'}$ and comparing the  prefix $\alpha{'}$ of $\textbf{v}$ and the suffix $\beta$ of $\textbf{x}$ when $|\alpha{'}|=|\beta|$.
Therefore, we conclude that $\textbf{x}$ can be appended to $S_{I,J}^{(k)}(n)$, namely $S_{I,J}^{(k)}(n)\bigcup\{\textbf{x}\}$ is  a cross-bifix-free code in this case.

{\bf Case 2)}  $k=n-2$. Take a specific element $\mathbf{x}\in I^2\times J \times I \times J^{n-4} $.
In this case, $S_{I,J}^{(k)}(n)=I^{n-2}\times J^2$. By using essentially  the same arguments as in the proof of {\bf Case 1)}, it is easy to show that $\textbf{x}$ can be appended to $S_{I,J}^{(n-2)}(n)$.

Therefore,  we can always  find an appropriate vector  $\textbf{x}=(x_1,\cdots,x_n)\in \mathbb{Z}_q^n\setminus S_{I,J}^{(k)}(n)$  appending to $S_{I,J}^{(k)}(n)$ such that $S_{I,J}^{(k)}(n)\bigcup\{\mathbf{x}\}$ is still a  cross-bifix-free code.  We are done.
\end{proof}

Lemma \ref{non-expandable1} says that  the $q$-ary   code $S_{I,J}^{(k)}(n)$
given as in  Construction \ref{S2} is non-expandable if $k=n-1$ or $1\leq k<n/2$.
By virtue of Lemma \ref{theorem1}, we improve Lemma \ref{non-expandable1} by showing that
$S_{I,J}^{(k)}(n)$ is non-expandable if and only if  $k=n-1$ or $1\leq k<n/2$, as formally stated below.

\begin{Theorem}\label{theorem2}
Let $n\geq5$ be a positive integer.
The $q$-ary cross-bifix-free code $S_{I,J}^{(k)}(n)$ given as in  Construction \ref{S2} is non-expandable if and only if $k=n-1$ or $1\leq k<n/2$.
\end{Theorem}

As $S_{q}^{(k)}(n)$ is a particular subclass of $S_{I,J}^{(k)}(n)$,
  the following corollary gives an improvement of \cite[Theorem 3.1]{Chee}.

\begin{Corollary}
Let $n\geq5$ be a positive integer.
The $q$-ary cross-bifix-free code $S_{q}^{(k)}(n)$ given as in Construction \ref{S1} is non-expandable if and only if $k=n-1$ or $1\leq k<n/2$.
\end{Corollary}

Equivalently, Theorem \ref{theorem2} tells us that  $S_{I,J}^{(k)}(n)$ is  expandable precisely when   $n/2\leq k\leq n-2.$
For $5\leq n\leq 6$, it is easy to   verify that $S_{I,J}^{(3)}(5)\cup I^2\times J\times I\times J$,
$S_{I,J}^{(3)}(6)\cup I^2\times J^2\times I\times J$
and $S_{I,J}^{(4)}(6)\cup I\times J\times I\times \mathbb{Z}_{q} \times J^2$ are non-expandable cross-bifix-free codes.
We now turn to consider the problem that how to  expand $S_{I,J}^{(k)}(n)$ to a non-expandable cross-bifix-free code in the case where
 $n\geq 7$ and
$n/2\leq k\leq n-2$. For this purpose, we introduce the following definition.
\begin{Construction}\label{V}
Let  $n\geq 7$ and $q>1$ be integers.
Let  $I$ and $J$ form a bipartition of $\mathbb{Z}_{q}$. Let $k,t$ and $m$ be integers satisfying $n/2\leq k\leq n-2$,   $t=max\{2,n-k-1\}$
and $t+1 \leq m\leq n$ with $m\neq n-k+t$.
For $m=t+1$, define $V_{I,J}^{(t)}(m)=I^t\times J$;  for $t+2\leq m<n-k+t$,
define $V_{I,J}^{(t)}(m)=I^t\times J\times \mathbb{Z}_{q}^{m-t-2}\times J$;
for $n-k+t<m\leq n$,  $V_{I,J}^{(t)}(m)$ is defined to be the set of all vectors
$(s_1,s_2,\cdots,s_m)\in\mathbb{Z}_{q}^m$     satisfying the following   property:
$(s_1,s_2,\cdots,s_t)\in I^t$,   $s_{t+1}\in J $,  $s_{n-k+t}\in I$ and $s_m\in J$.
\end{Construction}
At this point, several remarks are in order.
First,  $(s_{t+2},s_{t+3},\cdots,s_{m-1})$ must be $I^k$-free  since $m-t-2\leq n-t-2\leq k-1$.
Second, the cardinality of $V_{I,J}^{(t)}(n)$ is equal to $v_{I,J}^{(t)}(n)=|V_{I,J}^{(t)}(n)|=|I|^{t+1}|J|^2q^{n-t-3}$. Finally, $V_{I,J}^{(t)}(n)$ is, in general, not a cross-bifix-free code (see Example \ref{Example} below).
In order to obtain cross-bifix-free codes, we have to carefully pick out some vectors in  $V_{I,J}^{(t)}(n)$ which leads to the following construction.
\begin{Construction}\label{U}
Let  $n\geq 7$   and $q>1$ be integers.
Let  $I$ and $J$ form a bipartition of $\mathbb{Z}_{q}$. Let $k,t$ and $m$ be integers satisfying $n/2\leq k\leq n-2$,   $t=max\{2,n-k-1\}$
and $t+1 \leq m\leq n$ with $m\neq n-k+t$. Define $U_{I,J}^{(t)}(m)$ to be a subset of $V_{I,J}^{(t)}(m)$
satisfying the following  additional property:
The suffix of $\textbf{s}=(s_1,s_2,\cdots,s_m)\in V_{I,J}^{(t)}(m)$ with length $\ell$ does not belong to $V_{I,J}^{(t)}(\ell)$ for all  $t+1\leq \ell\leq m-(t+1)$ with $\ell\neq n-k+t$.
\end{Construction}
The next lemma says that $U_{I,J}^{(t)}(n)$ is a cross-bifix-free code. In fact, one can similarly  show  that  $U_{I,J}^{(t)}(m)$
is also a cross-bifix-free code; we do not state this more general result here because we merely use the cross-bifix-free property of $U_{I,J}^{(t)}(n)$ later.
\begin{lem}\label{lem1}
For any fixed $n\geq 7$, $n/2\leq k\leq n-2$ and $t=\max\{2,n-k-1\} $, the $q$-ary code $U_{I,J}^{(t)}(n)$ defined as in Construction \ref{U} is a cross-bifix-free code.
\end{lem}
\begin{proof}
We first show that each $\textbf{u}\in U_{I,J}^{(t)}(n)$ is bifix-free.
For any given $\textbf{u}\in U_{I,J}^{(t)}(n) $, it can be written as $\textbf{u}=\alpha\omega\beta$, where $\alpha$ and $\beta $ are necessarily non-empty strings, whereas $\omega$ can   be empty.

If $|\alpha|=|\beta|=j\leq t$,
then $\alpha\in I^{j}$, whereas the last coordinate of $\beta$
 belongs to $J$,
giving $\alpha\neq\beta$.

If $t+1\leq|\alpha|=|\beta|=j\leq n-(t+1)$ and $j\neq n-k+t$, we use the proof by contradiction.
Suppose otherwise that $\alpha$ and $\beta$ satisfy $\alpha=\beta$.
Since the last coordinate of $\beta$ is in $J$, then $\alpha= \beta\in V_{I,J}^{(t)}(j)$ contradicts to the hypothesis $\textbf{u}\in U_{I,J}^{(t)}(n)$, forcing $\alpha\neq\beta.$

If $|\alpha|=|\beta|=n-k+t$, then the last coordinate of $\alpha$   belongs to $I$,  whereas the last coordinate of $\beta$
 belongs to $J$. Thus, $\alpha\neq\beta.$

If $n-(t+1)<|\alpha|=|\beta|=j\leq n-1$, then the first $t+1$ coordinates  of $\alpha$ lie in  $I^t\times J$,  whereas
$\beta\in I^{j-(n-t)}\times J\times R(n-t-1),$ where $R(n-t-1)$
is the set of $q$-ary $U_{I,J}^{(t)}(m)$-free words of length
$n-t-1$ that end with an element of $J$ (here, $t+1\leq m\leq n-(t+1)$ and $m\neq t+n-k$).
Noting that $0\leq j-(n-t)\leq t-1$, we have
$\alpha\neq\beta$.

Therefore, we have shown that every $\textbf{u}\in U_{I,J}^{(t)}(n)$ is bifix-free.
The proof for $\textbf{u}$ and $\textbf{u}{'}$ being cross-bifix-free for any $\textbf{u}\neq\textbf{u}{'}\in U_{I,J}^{(t)}(n)$ is quite analogous to the previous one just illustrated. We omit it here.
\end{proof}

We now show that $S_{I,J}^{(k)}(n)\bigcup U_{I,J}^{(t)}(n)$ is a cross-bifix-free code, which is the conclusion of the next result.
\begin{lem}\label{lem2}
For any fixed $n\geq 7$, $n/2\leq k\leq n-2$ and $t=\max\{2,n-k-1\}$,  $S_{I,J}^{(k)}(n)\bigcup U_{I,J}^{(t)}(n)$ is a cross-bifix-free code.
\end{lem}
\begin{proof}
It follows from Lemma \ref{lem1} that $U_{I,J}^{(t)}(n)$ is a cross-bifix-free code.
Construction \ref{S2} tells us that $S_{I,J}^{(k)}(n)$ is a cross-bifix-free code.
To complete the proof we only need to show the cross-bifix-free property between $ S_{I,J}^{(k)}(n)$ and $U_{I,J}^{(t)}(n)$.
Take a typical element  $\textbf{s}\in S_{I,J}^{(k)}(n)$ and it can be written as $\textbf{s}=\alpha\omega$,
where $|\alpha|>0$ and $|\omega|>0$.
Similarly, take a typical element $\textbf{s}{'}\in U_{I,J}^{(t)}(n)$  and it can be written  as $s{'}=\omega{'}\beta{'}$, where $|\omega{'}|>0$ and  $|\beta{'}|>0$.

If $|\omega|=|\omega{'}|=i\leq t$, then $\omega{'}\in I^{i}$ and $\omega \in \mathbb{Z}_q^{i-1} \times J $, yielding $\omega{'}\neq \omega$.

If $t<|\omega|=|\omega{'}|=i<n-k+t$,
then the first $t+1$ coordinates  of $\omega{'}$ lie in  $I^t\times J$, whereas
$\omega \in I^{i-(n-k)}\times J \times\mathbb{Z}_q^{n-k-2}\times J$ (here we use the fact that $i>t\geq n-k-1$).
Since $t>i-(n-k)$,
 we have
$\omega{'}\neq \omega$.

If $|\omega|=|\omega{'}|=i=n-k+t$, then $\omega{'}\in I^{t}\times J \times \cdots \times I$ and $\omega \in I^{t}\times J \times\mathbb{Z}_q^{n-k-2}\times J.$
The last coordinate of $\omega{'}$   belongs to $I$,  whereas the last coordinate of $\omega$
 belongs to $J$.
 Since the last coordinate of $\omega$ is different from that of  $\omega{'}$,
we have $\omega{'}\neq \omega$.

If $n-k+t<|\omega|=|\omega{'}|=i\leq n-1$, then $\omega \in I^{i-(n-k)}\times J \times\mathbb{Z}_q^{n-k-2}\times J$ and
$\omega{'}\in I^{t}\times J \times \mathbb{Z}_q^{i-t-1}$.
It follows that $\omega \in I^{i-(n-k)}\times J \times\mathbb{Z}_q^{n-k-2}\times J$ and
the first $t+1$ coordinates  of $\omega{'}$ lie in  $I^t\times J$.
 Since $t<i-(n-k)$,  we have
$\omega{'}\neq \omega$.

To summarize,
no   prefix of $\textbf{s}{'}$ is a suffix of $\textbf{s}$ and vice versa for any $\textbf{s}{'}\in U_{I,J}^{(t)}(n)$ and
$\textbf{s}\in S_{I,J}^{(k)}(n)$.
Therefore, for any fixed $n\geq 7$, $\frac{n}{2}\leq k\leq n-2$ and $t=\max\{2,n-k-1\}$,  $ S_{I,J}^{(k)}(n)\bigcup U_{I,J}^{(t)}(n)$ is a cross-bifix-free code.
\end{proof}

We remain to show that $S_{I,J}^{(k)}(n)\bigcup U_{I,J}^{(t)}(n)$ is a non-expandable cross-bifix-free code.
As above, fix integers $n,n/2\leq k\leq n-2,t=\max\{2,n-k-1\}$ and $t+1\leq m\leq n-(t+1)$ with $m\neq n-k+t$.
In order to achieve our goal, we need to give an alternative presentation of  Construction \ref{U}.
Let $Q_m(n)$ denote the subset of $V_{I,J}^{(t)}(n)$ such that  the suffix of
$\textbf{s}\in V_{I,J}^{(t)}(n)$ of length $m$ belongs to $V_{I,J}^{(t)}(m)$.
In symbols,
\begin{align*}
Q_m(n)=\big\{\textbf{s}\in V_{I,J}^{(t)}(n)\,\big|\ \mbox{the suffix of  \textbf{s} of length $m$ belongs to $ V_{I,J}^{(t)}(m)$} \big\}.
\end{align*}
Similarly,
let $P_m(n)$ denote the subset of $V_{I,J}^{(t)}(n)$ such that the suffix of
$\textbf{s}\in V_{I,J}^{(t)}(n)$ of length $m$ belongs to $U_{I,J}^{(t)}(m)$, i.e.,
\begin{align*}
P_m(n)=\big\{\textbf{s}\in V_{I,J}^{(t)}(n)\,\big|\, \mbox{the suffix of  $\textbf{s}$  of length $m$ belongs to $ U_{I,J}^{(t)}(m)$}\big\}.
\end{align*}
We have the following result, which turns out to be useful in the proof of the non-expandable property of  $S_{I,J}^{(k)}(n)\bigcup U_{I,J}^{(t)}(n)$ and in the enumeration of the size of  $U_{I,J}^{(t)}(n)$ (see Section \ref{section4}).
\begin{lem}\label{PQ}
Let notation be the same as above.
For any fixed $n\geq 7$, $n/2\leq k\leq n-2$ and $t=max\{2,n-k-1\}$, the following equation holds
\begin{equation*}
\bigcup\limits_{{\scriptsize\begin{array}{c}m=t+1,\\m\neq n-k+t \end{array}}}^{n-(t+1)}Q_m(n)= \bigcup\limits_{{\scriptsize\begin{array}{c}m=t+1,\\m\neq n-k+t \end{array}}}^{n-(t+1)}P_m(n).
\end{equation*}
\end{lem}
\begin{proof}
By the very definitions of $V_{I,J}^{(t)}(m)$ and $U_{I,J}^{(t)}(m)$ (see  Constructions \ref{V} and \ref{U}),
$U_{I,J}^{(t)}(m)$ is a subset of $V_{I,J}^{(t)}(m)$.
It follows  that $P_m(n)\subseteq Q_m(n)$ for each integer $t+1\leq m\leq n-(t+1)$ with $m\neq n-k+t$, which
immediately gives
$$\bigcup\limits_{{\scriptsize\begin{array}{c}m=t+1,\\m\neq n-k+t \end{array}}}^{n-(t+1)}P_m(n)\subseteq \bigcup\limits_{{\scriptsize\begin{array}{c}m=t+1,\\m\neq n-k+t \end{array}}}^{n-(t+1)}Q_m(n).
$$
For the
reverse inclusion,
we only need to prove that for each $t+1\leq m\leq n-(t+1)$ with $m\neq n-k+t$,
\begin{equation}\label{pq}
Q_m(n)\subseteq \bigcup\limits_{{\scriptsize\begin{array}{c}r=t+1,\\r\neq n-k+t \end{array}}}^{n-(t+1)}P_r(n).
\end{equation}
For any fixed $t+1\leq m<n-k+t$, it is easy to see from the definition  of Construction \ref{U} that $V_{I,J}^{(t)}(m)=U_{I,J}^{(t)}(m)$
\big(since $m<n-k+t$ and $n-k-1\leq t$, we have $m<2(t+1)$. This means that there is no integer $\ell$ satisfying $t+1\leq\ell\leq m-(t+1)$\big). This gives  $Q_m(n)=P_m(n)$, proving Equation (\ref{pq}).
We use an induction argument on $m$ to show that Equation (\ref{pq}) holds true for   $t+1\leq m\leq n-(t+1)$ with $m\neq n-k+t$.
Assume that
Equation (\ref{pq}) is true for $m\leq n-(t+2)$.
Our idea is to separate $Q_m(n)$ into two disjoint subsets, each of which is contained in the right hand side of Equation (\ref{pq}).
For $m=n-(t+1)$,
we first define a subset $Q_m^1(n)$
of $Q_m(n)$ as follows: $Q_m^1(n)$ consists of elements  $\textbf{v}\in Q_m(n)$ satisfying  that there exists an integer $\ell$
with $t+1 \leq \ell\leq m-(t+1)=n-2(t+1)$ and $\ell\neq n-k+t$ such that
the suffix of $\textbf{v}$ of length $\ell$ belongs to $V_{I,J}^{(t)}(\ell)$.
It follows that
 \begin{equation*}
 Q_m^1(n)\subseteq \bigcup\limits_{{\scriptsize\begin{array}{c}\ell=t+1 \\ \ell\neq n-k+t \end{array}}}^{n-2(t+1)}Q_\ell(n).
 \end{equation*}
As $\ell\leq n-2(t+1)\leq n-(t+2)$, by the inductive hypothesis, we have
 \begin{equation*}
  \bigcup\limits_{{\scriptsize\begin{array}{c}\ell=t+1,\\ \ell\neq n-k+t \end{array}}}^{n-2(t+1)}Q_\ell(n)\subseteq \bigcup\limits_{{\scriptsize\begin{array}{c}r=t+1,\\r\neq n-k+t \end{array}}}^{n-(t+2)}P_r(n).
\end{equation*}
Hence, we have
$$
Q_m^1(n)\subseteq \bigcup\limits_{{\scriptsize\begin{array}{c}r=t+1,\\r\neq n-k+t \end{array}}}^{n-(t+2)}P_r(n).
$$
Now let
$Q_m^{2}(n)$
be the  complementary set of $Q_m^{1}(n)$ in $Q_m(n)$. It follows that
$Q_m^{2}(n)$ consists of elements   $\textbf{v}\in Q_m(n)\subseteq V_{I,J}^{(t)}(n)$ satisfying that
the  suffix of $\textbf{v}$ \big(whose $m$-length suffix belongs to $\in V_{I,J}^{(t)}(m)$\big) with length $\ell$ does not belong to $V_{I,J}^{(t)}(\ell)$
for all $t+1 \leq \ell\leq m-(t+1)=n-2(t+1)$ with $\ell\neq n-k+t$.
At this point, we claim that
$V_{I,J}^{(t)}(m)=U_{I,J}^{(t)}(m)$.
Indeed, fix an arbitrary element $\textbf{x}\in V_{I,J}^{(t)}(m)$.
By Construction \ref{V}, as $m=n-(t+1)$, an element ${\bf\widehat{x}}\in V_{I,J}^{(t)}(n)$ can be found such that the $m$-length suffix of $\bf\widehat{x}$
is equal to $\textbf{x}\in V_{I,J}^{(t)}(m)$.
We then have ${\bf\widehat{x}}\in Q_m(n)$.
If ${\bf\widehat{x}}\in Q_m^2(n)$, then the $m$-length suffix of ${\bf\widehat{x}}$ (i.e., $\textbf{x}$) satisfies that
the  suffix of $\textbf{x}$   with length $\ell$ does not belong to $V_{I,J}^{(t)}(\ell)$
for all $t+1 \leq \ell\leq m-(t+1)=n-2(t+1)$ with $\ell\neq n-k+t$.
We immediately get $\textbf{x}\in U_{I,J}^{(t)}(m)$. If  ${\bf\widehat{x}}\in Q_m^1(n)$,
then
$$
{\bf\widehat{x}}\in Q_m^1(n)\subseteq \bigcup\limits_{{\scriptsize\begin{array}{c}r=t+1,\\r\neq n-k+t \end{array}}}^{n-(t+2)}P_r(n).
$$
We again see that the $m$-length suffix of ${\bf\widehat{x}}$ (i.e., $\textbf{x}$) belongs to $U_{I,J}^{(t)}(m)$.
Thus, we have shown that $V_{I,J}^{(t)}(m)=U_{I,J}^{(t)}(m)$, as claimed.
Then we have $Q_{m}^2(n)=P_{m}(n)$.
It follows that
$$
Q_{m}^2(n)\subseteq \bigcup\limits_{{\scriptsize\begin{array}{c}r=t+1,\\r\neq n-k+t \end{array}}}^{n-(t+1)}P_r(n).
$$
Since
$Q_{m}(n)=Q_{m}^1(n)\bigcup Q_{m}^2(n),
$
we have
$$
Q_m(n)\subseteq \bigcup\limits_{{\scriptsize\begin{array}{c}r=t+1,\\r\neq n-k+t \end{array}}}^{n-(t+1)}P_r(n)
$$
for every $t+1 \leq m\leq n-(t+1)$ and $m\neq n-k+t$. The proof is  completed.
\end{proof}

According to Lemma \ref{PQ},  the additional property in Construction \ref{U} can   be rewritten as
``The suffix of $\textbf{s}$ with length $m$ does not belong to $U_{I,J}^{(t)}(m)$, where $t+1\leq m\leq n-(t+1)$ and $m\neq n-k+t$."
Construction \ref{U} is therefore equivalent to the following construction.
\vspace{0.2cm}

\noindent
{\bf Construction 3.5$'$}
Let  $n\geq 7$  and $q>1$ be integers.
Let  $I$ and $J$ form a bipartition of $\mathbb{Z}_{q}$. Let $k,t$ and $m$ be integers satisfying $n/2\leq k\leq n-2$,   $t=\max\{2,n-k-1\}$
and $t+1 \leq m\leq n$ with $m\neq n-k+t$. Define $U_{I,J}^{(t)}(m)$ to be a subset of $V_{I,J}^{(t)}(m)$
satisfying  the following  additional property:
The suffix of $\textbf{s}=(s_1,s_2,\cdots,s_m)\in V_{I,J}^{(t)}(m)$ with length $\ell$ does not belong to $U_{I,J}^{(t)}(\ell)$, where $t+1\leq \ell\leq m-(t+1)$ and $\ell\neq n-k+t$.

Note that the value of  $t=\max\{2,n-k-1\}$ in Construction  3.5$'$ is determined by the values of $n$ and $k$.
To correctly  understand Constructions  \ref{V}, \ref{U} and 3.5$'$,  we point out that the structures of
$V_{I,J}^{(t)}(m)$ and $U_{I,J}^{(t)}(m)$  actually  depend on the value of $k$, not just $t$.
For example,
$k=n-2$ and $k=n-3$ are resulting in  $t=2$; however,  the meaning for $V_{I,J}^{(2)}(n)$ and $U_{I,J}^{(2)}(n)$ is different according to various
values of $k$.
For $k=n-2$, the fourth coordinate of $V_{I,J}^{(2)}(n)$ and $U_{I,J}^{(2)}(n)$ must belong to $I$. For $k=n-3$, the fifth coordinate of $V_{I,J}^{(2)}(n)$ and $U_{I,J}^{(2)}(n)$ must belong to $I$.
We believe that after reading the context, there will be no confusion.
For convenience, we adopt the notations  $V_{I,J}^{(t)}(m)$ and $U_{I,J}^{(t)}(m)$.

We now include a small example to illustrate Construction 3.5$'$.
\begin{Example}\label{Example}{\rm
(1) Take $n=9, k=6 $ and $ q=2$. In this case, $t=2$, $3\leq m\leq 6$ and $m\neq5$.
By Construction \ref{V}, it is easy to see that
\begin{equation*}
\begin{split}
V_{\{0\}\{1\}}^{(2)}(9)=
&\big\{001000001, 001001001, 001000011,
001000101,001000111,001001011,001001101,001001111,\\&001100001,
001100011,001100101,
001100111,001101001,001101011,001101101,001101111\big\}.
\end{split}
\end{equation*}
Clearly, it is readily seen that $V_{\{0\}\{1\}}^{(2)}(9)$ is not a cross-bifix-free code.
Construction \ref{U} (or equivalently Construction $3.5'$) enables us to find a cross-bifix-free code within
$V_{\{0\}\{1\}}^{(2)}(9)$. Indeed, we have
$$
U_{\{0\}\{1\}}^{(2)}(3)=\big\{001\big\},~~~~
U_{\{0\}\{1\}}^{(2)}(4)=\big\{0011\big\}~~\hbox{and}~~ U_{\{0\}\{1\}}^{(2)}(6)=\big\{001101\big\}.
$$
By the additional property of Construction 3.5$'$,
it can be concluded that
\begin{equation*}
\begin{split}
U_{\{0\}\{1\}}^{(2)}(9)=&\big\{001000101,001000111,
001001011,001001111,001100101,001100111,001101011,\\
&001101101,
001101111\big\},
\end{split}
\end{equation*}
is a cross-bifix-free  code.

(2) Take $n=9, k=7 $ and $ q=2$. In this case, $t=2$, $3\leq m\leq 6$ and $m\neq4$.
By Construction \ref{V}, it is easy to see that
\begin{equation*}
\begin{split}
V_{\{0\}\{1\}}^{(2)}(9)=
&\big\{001000001,001000011 ,001000101 ,001000111
,001001001,001001011,001001101,001001111,\\&001010001,001010011
,001010101,001010111
,001011001,001011011,001011101,001011111\big\}.
\end{split}
\end{equation*}
Clearly,  $V_{\{0\}\{1\}}^{(2)}(9)$ is not a cross-bifix-free code.
Construction \ref{U} (or equivalently Construction $3.5'$) enables us to find a cross-bifix-free code within
$V_{\{0\}\{1\}}^{(2)}(9)$. Indeed, we have
$$
U_{\{0\}\{1\}}^{(2)}(3)=\big\{001\big\},~~~~
U_{\{0\}\{1\}}^{(2)}(5)=\big\{00101\big\}~~\hbox{and}~~ U_{\{0\}\{1\}}^{(2)}(6)=\big\{001011\big\}.
$$
By the additional property of Construction 3.5$'$,
we see that
\begin{equation*}
\begin{split}
U_{\{0\}\{1\}}^{(2)}(9)=&\big\{001000011,001000111
,001001101,001001111,001010011,001010101,001010111,001011011,\\
&001011101,001011111
\big\}
\end{split}
\end{equation*}
is a cross-bifix-free  code.}
\end{Example}

We have shown in Lemma \ref{lem2} that $S_{I,J}^{(k)}(n)\bigcup U_{I,J}^{(t)}(n)$ is a cross-bifix-free code.
With Lemma \ref{lem2} and the equivalence of Constructions \ref{U} and   3.5$'$, we are now in a position to obtain   that
$S_{I,J}^{(k)}(n)\bigcup U_{I,J}^{(t)}(n)$ is indeed non-expandable.
\begin{Theorem}
For any fixed $n\geq 7$, $n/2\leq k\leq n-2$ and $t=\max\{2,n-k-1\}$,  $S_{I,J}^{(k)}(n)\bigcup U_{I,J}^{(t)}(n)$ is a
non-expandable cross-bifix-free code.
\end{Theorem}
\begin{proof}
Take a typical $q$-ary vector $\mathbf{x}=(x_1,x_2,\cdots,x_n)\in \mathbb{Z}_q^n\setminus\big(S_{I,J}^{(k)}(n)\bigcup U_{I,J}^{(t)}(n)\big)$.  It is enough to show that $S_{I,J}^{(k)}(n)\bigcup U_{I,J}^{(t)}(n)\bigcup\big\{\mathbf{x}\big\}$ is not a cross-bifix-free code, or in other words,
$\mathbf{x}$ overlaps with  some codeword in $S_{I,J}^{(k)}(n)\bigcup U_{I,J}^{(t)}(n)$.
Clearly, the $q$-ary vector $\mathbf{x}=(x_1,x_2,\cdots,x_n)$ overlaps with every codeword in $S_{I,J}^{(k)}(n)\bigcup U_{I,J}^{(t)}(n)$ if $x_1\in J$ or $x_n\in I$.
Therefore, we only need to consider the case where  $\mathbf{x} \in (I \times \mathbb{Z}_{q}^{n-2} \times J) \setminus \big(S_{I,J}^{(k)}(n)\bigcup U_{I,J}^{(t)}(n)\big)$.
Now we show that $\textbf{x}$ can be restricted  to be $I^k$-free.
To this end, suppose otherwise that $\textbf{x}$ is not $I^k$-free.
In this case, $\textbf{x}$ must be in one of the following two forms:
$\mathbb{Z}_q^{n-k-1}\times I^k\times J$ and
$\mathbb{Z}_q^{n-k-1-r}\times I^k\times J\times R(r)$,
where $R(r)$ is an $I^k$-free string with length $r$  of which the last coordinate belongs to $J$.
If $\textbf{x}$ is in the first form, then $\textbf{x}$ overlaps with $S_{I,J}^{(k)}(n)$ because the $(k+1)$-length
suffix of  $\textbf{x}$ is identical with the $(k+1)$-length prefix of some codeword of $S_{I,J}^{(k)}(n)$.
For the case where $\textbf{x}$ is in the second form,
note that $I^k\times J\times R(r)\times J^{n-k-r-1}\subseteq S_{I,J}^{(k)}(n)$.
It follows that  $\textbf{x}$ also overlaps with some codeword in $S_{I,J}^{(k)}(n)$ by considering the
$(k+1+r)$-length suffix  of $\textbf{x}$ and the codewords in $I^k\times J\times R(r)\times J^{n-k-r-1}$.
In conclusion, we can assume that $\mathbf{x} \in (I \times \mathbb{Z}_{q}^{n-2} \times J) \setminus \big(S_{I,J}^{(k)}(n)\bigcup U_{I,J}^{(t)}(n)\big)$ and $\textbf{x}$  is $I^k$-free.

Recall that $n/2\leq k\leq n-2$ and $t=\max\{2,n-k-1\} $.
As we did in the proof of  Lemma \ref{theorem1}, we consider the following two cases separately.

{\bf Case 1)} $n/2\leq k< n-2$.
Since $\mathbf{x}$   is $I^k-$free,
assume that $\mathbf{x}\in I^\ell \times J \times \mathbb{Z}_{q}^{n-\ell-1}$ for some $1\leq \ell\leq k-1$.
At this point, we  discuss the following subcases.
\begin{itemize}
\item[(1.1)] $\ell< n-k-1$. In this subcase,  the $(\ell+1)$-length prefix of $\mathbf{x}$ is identical with the
$(\ell+1)$-length suffix of some $\mathbf{y}$ in $I^k\times J^{n-k-\ell-1}\times I^\ell\times J\subseteq S_{I,J}^{(k)}(n)$. Hence, $\mathbf{x}$ overlaps with  some codeword from $S_{I,J}^{(k)}(n)\bigcup U_{I,J}^{(t)}(n)$.

\item[(1.2)] $n-k-1\leq \ell\leq k-1$ and the $(n-k+\ell)$-th coordinate of $\mathbf{x}$ belongs to $J$. In this subcase,
$\mathbf{x}$ is in the form $I^\ell\times J\times \mathbb{Z}_q^{n-k-2}\times J\times \mathbb{Z}_q^{n-\ell-(n-k)}$ and
the $(n-k+\ell)$-length prefix of $\mathbf{x}$ is
in the form $I^\ell\times J\times \mathbb{Z}_q^{n-k-2}\times J$.
On the other hand, the code $S_{I,J}^{(k)}(n)$ contains elements in the form
$I^{k-\ell}\times I^\ell \times J \times \mathbb{Z}_{q}^{n-k-2} \times J$.
We obtain the desire result by considering the $(n-k+\ell)$-length prefix of $\mathbf{x}$
and the $(n-k+\ell)$-length suffix of the codewords in $I^{k-\ell}\times I^\ell \times J \times \mathbb{Z}_{q}^{n-k-2} \times J$.

\item[(1.3)] $n-k-1\leq \ell\leq k-1$ and the $(n-k+\ell)$-th coordinate  of $\mathbf{x}$ belongs to $I$.
In this subcase, $\textbf{x}$ is in the form $I^{\ell-t}\times I^{t}\times J\times\cdots\times I\times\cdots \times J$,
where the single $I$ is in the $(n-k+\ell)$-th coordinate.
Note that $I^{t}\times J\times\cdots\times I\times\cdots \times J$ is exactly equal to
$V_{I,J}^{(t)}\big(n-(\ell-t)\big)=V_{I,J}^{(t)}\big(n-\ell+t\big)$ (see Construction \ref{V}).
We have $\textbf{x}\in I^{\ell-t}\times V_{I,J}^{(t)}\big(n-\ell+t\big)$.
Assume first that  the $m$-length suffix  of $\mathbf{x}$ does not belong to $U_{I,J}^{(t)}(m)$ for all
$t+1\leq m \leq n-\ell+t-(t+1)=n-(\ell+1)\leq n-(t+1)$ and $m\neq n-k+t$.
In this situation,
$\mathbf{x}$ belongs to $I^{\ell-t}\times U_{I,J}^{(t)}\big(n-\ell+t\big)$
since $\ell-t>0$. This implies that
the $(n-\ell+t)$-length  suffix of  $\mathbf{x}$ is a prefix of some $\mathbf{y}$ in
$ U_{I,J}^{(t)}(n-\ell+t)\times J^{\ell-t}\subseteq U_{I,J}^{(t)}(n)$.
 We  therefore remain to consider the case where
there exists an integer $m$ such that the $m$-length suffix of $\mathbf{x}$  belongs to $U_{I,J}^{(t)}(m)$,
where $t+1\leq m \leq n-(\ell+1)\leq n-(t+1)$ and $m\neq n-k+t$.
The $m$-length suffix of $\mathbf{x}$ is thus identical with an $m$-length prefix of some element in
$ U_{I,J}^{(t)}(m)\times\mathbb{Z}_{q}^{n-k+t-m-1} \times I\times J^{k-t} \subseteq U_{I,J}^{(t)}(n)$ $(t+1\leq m<n-k+t)$ or
$ U_{I,J}^{(t)}(m)\times J^{n-m} \subseteq U_{I,J}^{(t)}(n)$ $(t+n-k<m\leq n-(\ell+1)\leq n-(t+1))$.
Hence, $\mathbf{x}$ overlaps with  some codeword from $S_{I,J}^{(k)}(n)\bigcup U_{I,J}^{(t)}(n)$.
\end{itemize}

{\bf Case 2)}
$ k=n-2$. In this case, let $\mathbf{x}\in I^\ell \times J \times \mathbb{Z}_{q}^{n-\ell-1}$ for some $1\leq \ell\leq n-3$. At this point, we  discuss the following subcases.
\begin{itemize}
\item[(2.1)] $\ell< t=2$. In this subcase,  the $(\ell+1)$-length (or $2$-length) prefix of $\mathbf{x}$ is identical with the
$(\ell+1)$-length (or $2$-length) suffix of some $\mathbf{y}$ in $I^2\times J\times I\times J^{n-6}\times I\times J \subseteq U_{I,J}^{(2)}(n)$ since $n-6>0$. Hence, $\mathbf{x}$ overlaps with  some codeword from $S_{I,J}^{(k)}(n)\bigcup U_{I,J}^{(t)}(n)$.
\end{itemize}

The proofs for the subcases (2.2) and (2.3) are quite similar to the subcases (1.2) and (1.3) of {\bf Case 1)}, where
the subcase (2.2) is read as $2=t\leq \ell\leq k-1=n-3$ and the $(n-k+\ell)$-th (i.e., $(2+\ell)$-th) coordinate of $\mathbf{x}$ belongs to $J$,
and the subcase (2.3) is read as  $2=t\leq \ell\leq k-1=n-3$ and the $(n-k+\ell)$-th (i.e., $(2+\ell)$-th) coordinate  of $\mathbf{x}$ belongs to $I$.
We omit the proofs here.

Therefore, we conclude that
$S_{I,J}^{(k)}(n)\bigcup U_{I,J}^{(t)}(n)\bigcup\big\{\mathbf{x}\big\}$ is not a cross-bifix-free code for all $\mathbf{x}\in \mathbb{Z}_q^n\setminus\big(S_{I,J}^{(k)}(n)\bigcup U_{I,J}^{(t)}(n)\big)$.  The proof is thus completed.
\end{proof}

\section{Count the size of  $U_{I,J}^{(t)}(n)$}\label{section4}
We have just shown that  $S_{I,J}^{(k)}(n)\bigcup U_{I,J}^{(t)}(n)$ is a
non-expandable cross-bifix-free code if $n\geq 7$, $n/2\leq k\leq n-2$ and $t=\max\{2,n-k-1\}$.
In this section, we assume that $n\geq 7$, $n/2\leq k\leq n-2$ and $t=\max\{2,n-k-1\}$.
Note that   $S_{I,J}^{(k)}(n)\bigcup U_{I,J}^{(t)}(n)$ is a  disjoint union
simply because $k\neq t$.
Note also that $S_{I,J}^{(k)}(n)=I^k\times J\times \mathbb{Z}_q^{n-k-2}\times J$ since $n-k-2\leq k-1$ by our assumption $k\geq n/2$.
We immediately have $|S_{I,J}^{(k)}(n)|=q^{n-k-2}|I|^k|J|^2$.
Consequently, the problem of counting the size of $S_{I,J}^{(k)}(n)\bigcup U_{I,J}^{(t)}(n)$  is fully converted to that of
counting the size of $U_{I,J}^{(t)}(n)$.
The purpose of this section is to present an enumerative formula for the size of $U_{I,J}^{(t)}(n)$.

Recall that
\begin{align*}
P_m(n)=\big\{\textbf{s}\in V_{I,J}^{(t)}(n)\,\big|\, \mbox{the suffix of  $\textbf{s}$  of length $m$ belongs to $ U_{I,J}^{(t)}(m)$}\big\},
\end{align*}
where $m$ is a positive integer satisfying $t+1\leq m\leq n-(t+1)$ and $m\neq n-k+t$.
To count the size of $U_{I,J}^{(t)}(n)$, we need the following lemma.
\begin{lem}\label{disjoint}
Let notation be the same as above. We have
$P_m(n)\bigcap P_\ell(n)=\emptyset$, where $m,\ell$ satisfy
$t+1\leq m<\ell\leq n-(t+1)$ and $m,\ell\neq n-k+t$.
\end{lem}
\begin{proof}
We consider the following two cases separately.
\begin{itemize}

\item[\bf{Case 1)}] $t+1\leq m<\ell-(t+1)$ and    $m,\ell\neq n-k+t$.
Suppose otherwise that $P_m(n)\bigcap P_\ell(n)$ is non-empty, say $\textbf{u}\in P_m(n)\bigcap P_\ell(n)$.
It follows from $\textbf{u}\in P_\ell(n)$ that the $\ell$-length suffix of $\textbf{u}$ belongs to $U_{I,J}^{(t)}(\ell)$.
By the definition of $U_{I,J}^{(t)}(\ell)$ (see Construction $3.5'$), the $m$-length suffix of every vector in $U_{I,J}^{(t)}(\ell)$
does not belong to $U_{I,J}^{(t)}(m)$ for any $m$ satisfying $t+1\leq m<\ell-(t+1)$ and    $m\neq n-k+t$.
This implies that the $m$-length suffix of $\textbf{u}$ does not belong to $U_{I,J}^{(t)}(m)$, contradicting to the assumption
$\textbf{u}\in P_m(n)$. We thus have shown that $P_m(n)\bigcap P_\ell(n)=\emptyset$.

\item[\bf{Case 2)}] $\ell-(t+1)\leq m<\ell$ and    $m,\ell\neq n-k+t$.
Suppose otherwise that $P_m(n)\bigcap P_\ell(n)$ is non-empty, say $\textbf{u}\in P_m(n)\bigcap P_\ell(n)$.
As in the previous case,  the $\ell$-length suffix of $\textbf{u}$ belongs to $U_{I,J}^{(t)}(\ell)$.
Recall that $U_{I,J}^{(t)}(\ell)$ is a subset of
$V_{I,J}^{(t)}(\ell)=I^t\times J\times \mathbb{Z}_q^{\ell-t-2}\times J$.
If $m=\ell-t$, then  the $m$-length suffix of $\textbf{u}$ belongs to  $J\times \mathbb{Z}_q^{\ell-t-2}\times J$
which is disjoint from $U_{I,J}^{(t)}(m)$. This gives that the $m$-length suffix of $\textbf{u}$  does not belong to $U_{I,J}^{(t)}(m)$,
or equivalently $\textbf{u}\notin P_m(n)$.
If $m>\ell-t$, then  the $m$-length suffix of $\textbf{u}$ belongs to  $I^{m-(\ell-t)}\times J\times \mathbb{Z}_q^{\ell-t-2}\times J$
which is also disjoint from $U_{I,J}^{(t)}(m)$ because $m-(\ell-t)<t$.
We then have $\textbf{u}\notin P_m(n)$, a contradiction again.
\end{itemize}
\end{proof}
Note that it is very easy to derive an explicit formula for the value of
$|U_{I,J}^{(t)}(n)|=u_{I,J}^{(t)}(n)$ if  $n=7$ and $k=n-2$.
Indeed,
if $ n=7$ then $k=n-2=5$ and
$
u_{I,J}^{(t)}(n)=q^{2}|I|^{3}|J|^2-|I|^{2}|J||I|u_{I,J}^{(t)}(3)=|I|^{3}|J|^2(q^2-|I|^2).
$
Therefore, we only need to consider the case $n\geq8$ if $k=n-2$.
We are now in a position to state the main result of this section.

\begin{Theorem}\label{u}
Let notation be the same as above. We assume that $n\geq8$ if $k=n-2$.
Let $U_{I,J}^{(t)}(n)$ be the code given as in Construction $3.5'$.
The cardinality of  $U_{I,J}^{(t)}(n)$, denoted by $u_{I,J}^{(t)}(n)$,  is given as follows:

\noindent
If  $\frac{n}{2}\leq k<\frac{2}{3}n-\frac{1}{3}$, then
$$
u_{I,J}^{(t)}(n)=|I|^{n-k}|J|^2q^{k-2}-|I|^{2(n-k)-2}|J|^2q^{2k-n-1}\big((2k-n)|J|+q\big).
$$
\noindent
If  $k=\frac{2}{3}n-\frac{1}{3}$, then
$$u_{I,J}^{(t)}(n)=q^{k-2}|I|^{n-k}|J|^2-|I|^{2n-2k-2}|J|^2q^{2k-n-1}\big(|I|+(2k-n-1)|J|\big).$$
\noindent
If  $\frac{2}{3}n\leq k<\frac{3}{4}n-\frac{1}{2}$, then
\begin{equation*}
\begin{split}
u_{I,J}^{(t)}(n)&=q^{k-2}|I|^{n-k}|J|^2-|I|^{2n-2k-2}|J|^2q^{2k-n-2}\big((6k-4n+2)|I||J|+|I|q+(3n-3-4k)|J|q\big)\\
&+|I|^{3n-3k-3}|J|^3q^{3k-2n-1}\big((3k-2n+1)q+\frac{(3k-2n+1)(3k-2n)}{2}|J|\big).
\end{split}
\end{equation*}

\noindent
If  $\frac{3}{4}n-\frac{1}{2}\leq k<n-2$, then
\begin{equation*}
\begin{split}
u_{I,J}^{(t)}(n)&=q^{k-2}|I|^{n-k}|J|^2-\sum\limits_{{\scriptsize\begin{array}{c}m=t+1,\\m\neq t+n-k \end{array}}}^{k-t}
|I|^{t+1}|J|q^{n-t-m-2}u_{I,J}^{(t)}(m)\\
&-\sum\limits_{m=k-t+1}^{n-(t+1)}|I|^t|J|q^{n-t-m-1}u_{I,J}^{(t)}(m).
\end{split}
\end{equation*}

\noindent
If  $k=n-2$ and $n\geq 8$, then
$$
u_{I,J}^{(t)}(n)=q^{n-5}|I|^{3}|J|^2-\sum\limits_{{\scriptsize\begin{array}{c}m=3,\\m\neq 4 \end{array}}}^{n-4}
|I|^{3}|J|q^{n-m-4}u_{I,J}^{(2)}(m)-|I|^2|J|u_{I,J}^{(2)}(n-3).
$$
\end{Theorem}
\begin{proof}
By Construction $3.5'$, the cardinality of  $ U_{I,J}^{(t)}(n)$, denoted by $u_{I,J}^{(t)}(n)$, is equal to
$|V^{(t)}_{I,J}(n)|$ minus the size of the subset of vectors in
$V_{I,J}^{(t)}(n)$ whose $m$-length
suffixes belong  to $U_{I,J}^{(t)}(m)$ for every positive integer $m$ with $t+1 \leq m\leq n-(t+1)$ and $m\neq n-k+t$.
For convenience, for any fixed integer $m$ with $t+1 \leq m\leq k-t$ and $m\neq n-k+t$, the subset of vectors in $V_{I,J}^{(t)}(n)$ whose $m$-length
suffixes belong  to $U_{I,J}^{(t)}(m)$  is denoted by
\begin{equation*}
\label{eq1}
A_m(n)=\big\{\big(\alpha,\beta,U_{I,J}^{(t)}(m)\big)\,\big|\,\alpha\in I^t\times J\times\mathbb{Z}_q^{n-k-2}\times I,\beta\in\mathbb{Z}_q^{k-m-t}\big\};
\end{equation*}
for any fixed integer $m$ with $k-t<m\leq n-(t+1)$ and $m\neq n-k+t$, the subset of vectors in $V_{I,J}^{(t)}(n)$ whose $m$-length
suffixes belong  to $U_{I,J}^{(t)}(m)$  is denoted by
\begin{equation*}
\label{eq2}
B_m(n)=\big\{\big(\alpha,\beta,U_{I,J}^{(t)}(m)\big)\,\big|\,\alpha\in I^t\times J,\beta\in\mathbb{Z}_q^{n-t-1-m}\big\}.
\end{equation*}
Lemma \ref{disjoint} tells us that the sets
$A_m(n),A_{m'}(n)$,  $B_m(n)$ and $B_{m'}(n)$  are mutually disjoint for distinct integers  $m$ and $m'$ in the range
$t+1\leq m, m'\leq n-(t+1)$ and $m,m'\neq n-k+t$.

We obtained the desired result by dividing the values of $k$ into $4$ cases.

{\bf Case 1).} $\frac{n}{2}\leq k<\frac{2}{3}n-\frac{1}{3}$.  In this case $t=n-k-1$ and $k-t<t+1\leq n-(t+1)<n-k+t$.
By the definition of $A_m(n)$ and the fact $k-t<t+1$, one has $A_m(n)=\emptyset$. Using the fact $k-t<t+1$ again,
it is readily seen that $B_m(n)=\emptyset$ for $k-t<m<t+1$.
It follows from  Construction $3.5'$ that
\begin{align*}
u_{I,J}^{(t)}(n)&=|V_{I,J}^{(t)}(n)|-\sum\limits_{{\scriptsize\begin{array}{c}m=t+1, \end{array}}}^{n-(t+1)}|B_m(n)|\\
&=q^{n-t-3}|I|^{t+1}|J|^2-\sum\limits_{{\scriptsize\begin{array}{c}m=t+1,\\ \end{array}}}^{n-(t+1)}|I|^t|J|q^{n-t-m-1}u_{I,J}^{(t)}(m)\\
&=q^{n-t-3}|I|^{t+1}|J|^2-|I|^{2t}|J|^2q^{2k-n}-\sum\limits_{{\scriptsize\begin{array}{c}m=t+2,\\ \end{array}}}^{n-(t+1)}|I|^t|J|q^{n-t-m-1}|I|^t|J|q^{m-t-2}|J|\\
&=|I|^{n-k}|J|^2q^{k-2}-|I|^{2(n-k)-2}|J|^2q^{2k-n-1}\big((2k-n)|J|+q\big),
\end{align*}
where the last but one equality holds because
$u_{I,J}^{(t)}(t+1)=|I|^t|J|$ and $u_{I,J}^{(t)}(m)=|I|^t|J|q^{m-t-2}|J|$ for   $t+2 \leq m\leq n-(t+1)<n-k+t$.

{\bf Case 2).} $\frac{2}{3}n-\frac{1}{3}\leq k<\frac{3}{4}n-\frac{1}{2}$.  In this case $t=n-k-1$
and $t+1\leq k-t<n-k+t\leq n-(t+1)$.  Based on Construction $3.5'$, we have

\begin{align*}
u_{I,J}^{(t)}(n)&=|V_{I,J}^{(t)}(n)|-\sum\limits_{m=t+1}^{k-t}|A_m(n)|-\sum\limits_{{\scriptsize\begin{array}{c}m=k-t+1,\\m\neq n-k+t \end{array}}}^{n-(t+1)}|B_m(n)| \\
&=q^{n-t-3}|I|^{t+1}|J|^2-\sum\limits_{m=t+1}^{k-t}|I|^{t+1}|J|q^{n-t-m-2}u_{I,J}^{(t)}(m)-
\sum\limits_{{\scriptsize\begin{array}{c}m=k-t+1,\\m\neq n-k+t \end{array}}}^{n-(t+1)}|I|^t|J|q^{n-t-m-1}u_{I,J}^{(t)}(m).
\end{align*}
At this point, we consider two subcases:
\begin{itemize}
\item $k=\frac{2}{3}n-\frac{1}{3}$. We have  $k-t=t+1$ and  $n-(t+1)=n-k+t$. Hence,
\begin{align*}
u_{I,J}^{(t)}(n)&= q^{n-t-3}|I|^{t+1}|J|^2-|I|^{t+1}|J|q^{n-(t+1)-(t+2)}u_{I,J}^{(t)}(t+1)-
\sum\limits_{{\scriptsize\begin{array}{c}m=t+2 \end{array}}}^{n-(t+2)}|I|^t|J|q^{n-t-m-1}u_{I,J}^{(t)}(m)\\
&=q^{n-t-3}|I|^{t+1}|J|^2-|I|^{t+1}|J|q^{n-2t-3}|I|^t|J|-\sum\limits_{{\scriptsize\begin{array}{c}m=t+2 \end{array}}}^{n-(t+2)}|I|^t|J|q^{n-t-m-1}|I|^t|J|^2q^{m-t-2}\\
&=q^{k-2}|I|^{n-k}|J|^2-|I|^{2n-2k-2}|J|^2q^{2k-n-1}\big(|I|+(2k-n-1)|J|\big),
\end{align*}
where the last but one equality holds because
$u_{I,J}^{(t)}(t+1)=|I|^t|J|$ and  $u_{I,J}^{(t)}(m)=|I|^t|J|q^{m-t-2}|J|$ for $t+2 \leq m<n-k+t$.

\item $\frac{2}{3}n \leq k <\frac{3}{4}n-\frac{1}{2}$. We have $t+2 \leq k-t<t+n-k+1<n-(t+1)$ and
$u_{I,J}^{(t)}(n)$ is equal to
\begin{align*}
&q^{n-t-3}|I|^{t+1}|J|^2-\sum\limits_{m=t+1}^{k-t}|I|^{t+1}|J|q^{n-t-m-2}u_{I,J}^{(t)}(m)-
\sum\limits_{{\scriptsize\begin{array}{c}m=k-t+1,\\m\neq n-k+t \end{array}}}^{n-(t+1)}|I|^t|J|q^{n-t-m-1}u_{I,J}^{(t)}(m)\\
&=q^{k-2}|I|^{n-k}|J|^2-\sum\limits_{m=t+1}^{k-t}|I|^{t+1}|J|q^{n-t-m-2}u_{I,J}^{(t)}(m)-
\sum\limits_{{\scriptsize\begin{array}{c}m=k-t+1, \end{array}}}^{n-k+t-1}|I|^t|J|q^{n-t-m-1}u_{I,J}^{(t)}(m)\\
& -\sum\limits_{{\scriptsize\begin{array}{c}m=n-k+t+1, \end{array}}}^{n-(t+1)}|I|^t|J|q^{n-t-m-1}u_{I,J}^{(t)}(m)\\
&=q^{k-2}|I|^{n-k}|J|^2-|I|^{2n-2k-2}|J|^2q^{2k-n-2}((6k-4n+2)|I||J|+|I|q+(3n-3-4k)|J|q)\\
&+|I|^{3n-3k-3}|J|^3q^{3k-2n-1}((3k-2n+1)q+\frac{(3k-2n+1)(3k-2n)}{2}|J|),
\end{align*}
where the last but one equality holds because   $u_{I,J}^{(t)}(t+1)=|I|^t|J|$,
 $u_{I,J}^{(t)}(m)=|I|^t|J|q^{m-t-2}|J|$ for $t+2 \leq m<t+n-k$,
 and the values of $u_{I,J}^{(t)}(m)$ satisfies {\bf Case 1)} for   $n-k+t+1\leq m \leq n-(t+1)$.
\end{itemize}

{\bf Case 3).}  $\frac{3}{4}n-\frac{1}{2}\leq k<n-2$.
In this case    $n>6$, $t=n-k-1$ and $t+1<t+n-k\leq k-t<n-(t+1)$.  We have
\begin{align*}
 u_{I,J}^{(t)}(n) &=|V_{I,J}^{(t)}(n)|-\sum\limits_{{\scriptsize\begin{array}{c}m=t+1,\\m\neq t+n-k \end{array}}}^{k-t}|A_m(n)|-\sum\limits_{m=k-t+1}^{n-(t+1)}|B_m(n)| \\
& =q^{k-2}|I|^{n-k}|J|^2-\sum\limits_{{\scriptsize\begin{array}{c}m=t+1,\\m\neq t+n-k \end{array}}}^{k-t}
|I|^{t+1}|J|q^{n-t-m-2}u_{I,J}^{(t)}(m)- \\
& \sum\limits_{m=k-t+1}^{n-(t+1)}|I|^t|J|q^{n-t-m-1}u_{I,J}^{(t)}(m).
\end{align*}


{\bf Case 4).}  $k=n-2$ and $n\geq 8$.  In this case $t=2$ and $t+1<n-k+t\leq k-t<n-(t+1)$. We have
\begin{align*}
u_{I,J}^{(t)}(n)&=|V_{I,J}^{(t)}(n)|-\sum\limits_{{\scriptsize\begin{array}{c}m=3,\\m\neq 4 \end{array}}}^{n-4}|A_m(n)|-\sum\limits_{m=n-3}^{n-3}|B_m(n)|\\
& =q^{n-5}|I|^{3}|J|^2-\sum\limits_{{\scriptsize\begin{array}{c}m=3,\\m\neq 4 \end{array}}}^{n-4}
|I|^{3}|J|q^{n-m-4}u_{I,J}^{(2)}(m)- |I|^2|J|u_{I,J}^{(2)}(n-3).
\end{align*}
\end{proof}


We close this section by listing the values of $|S_2^{(k)}(n)|$ and $|S_2^{(k)}(n)\bigcup U_{2}^{(t)}(n)|$
in Tables $1$ and $2$ in the case where $n\geq5$ and $n/2\leq k\leq n-2$, respectively.
Through the tabulation form, one can easily see how many larger is $|S_2^{(k)}(n)\bigcup U_{2}^{(t)}(n)|$
than $|S_2^{(k)}(n)|$.
We take $n=17$ as an example to illustrate how to read Table $1$ (similarly Table $2$). First, fix an integer $n$; we take $n=17$.
Substituting $n$ into $n/2\leq k\leq n-2$, we see $9\leq k\leq15$. The number ``$64$" in the row index by ``$17$" and the column index by ``$9$"
means $|S_2^{(9)}(17)|=64.$ As an explicit example, we see that $|S_2^{(15)}(17)|=1$ and
$|S_2^{(15)}(17)\bigcup U_{2}^{(2)}(17)|=1433$.

\begin{table}[H]
    \caption{Cardinalities of  $S_2^{(k)}(n)$ ($5\leq n\leq 17, n/2\leq k\leq n-2$)}
    \begin{center}
    \begin{tabular}{c|p{18pt}p{18pt}p{18pt}p{18pt}p{18pt}p{18pt}p{18pt}p{18pt}p{18pt}p{18pt}p{18pt}p{18pt}p{18pt}}
        \hline\hline
        \diagbox{$n$}{$k$} & 3 & 4 & 5&6&7&8&9&10&11&12&13&14&15 \\
        \hline
        5 & 1 & & & & & & & & & & & & \\

        6 & 2& 1& & & & & & & & & & & \\

        7 &  & 2& 1&& & & & & & & & & \\

        8 &  & 4& 2& 1& & & & & & & & & \\

        9 &  & &4 & 2& 1& & & & & & & & \\

        10 &  & & 8& 4& 2& 1& & & & & & & \\

        11 &  & & &  8& 4& 2&1 & & & & && \\

        12 &  & & & 16& 8& 4& 2& 1& & & & & \\

        13 &  & & & &16& 8& 4& 2& 1& & & &  \\

        14 &  & & & &32&16& 8& 4& 2& 1& & &  \\

        15 &  & & & &&32&16& 8& 4& 2& 1& &   \\

        16 &  & & & &&64&32&16& 8& 4& 2& 1&   \\

        17 &  & & & &&&64&32&16& 8& 4& 2& 1\\
        \hline\hline
        \end{tabular}
    \end{center}
\end{table}

\begin{table}[H]
    \caption{Cardinalities of  $S_2^{(k)}(n)\bigcup U_{2}^{(t)}(n)$ ($5\leq n\leq 17, n/2\leq k\leq n-2$)  }
    \begin{center}
    \begin{tabular}{c|p{18pt}p{18pt}p{18pt}p{18pt}p{18pt}p{18pt}p{18pt}p{18pt}p{18pt}p{18pt}p{18pt}p{18pt}p{18pt}}
        \hline\hline
        \diagbox{$n$}{$k$} & 3 & 4 & 5&6&7&8&9&10&11&12&13&14&15 \\
        \hline
        5 & 2 & & & & & & & & & & & & \\

        6 & 3 &  2& & & & & & & & & & & \\

        7 &  & 3&  4& & & & & & & & & & \\

        8 &  & 7& 6& 6& & & & & & & & & \\

        9 &  & & 9& 11& 11& & & & & & & & \\

        10 &  & & 15& 12& 19& 19& & & & & & & \\

        11 &  & & & 21& 24& 34& 35& & & & & & \\

        12 &  & & & 31& 32& 45& 59& 64& & & & & \\

        13 &  & & & & 45& 52& 89& 107& 119& & & & \\

        14 &  & & & & 63& 72& 104& 166& 198& 221& & & \\

        15 &  & & & & & 93& 124& 201& 320& 371& 412& & \\

        16 &  & & & & & 127& 152& 224& 397& 615& 699& 768& \\

        17 &  & & & & & & 189& 268& 448& 794& 1173 &  1314 &   1433\\
        \hline\hline
        \end{tabular}
    \end{center}
\end{table}
\section{Concluding remarks and future work}
One of the main research problems on cross-bifix-free codes is to construct cross-bifix-free codes as large as possible in size.
In this paper, we further study the cross-bifix-free codes $S_{I,J}^{(k)}(n)$.
We improve the results  in \cite[Theorem 3.1]{Chee} and \cite[Theorem 1]{Wang} by
showing that $S_{I,J}^{(k)}(n)$ is  expandable precisely when   $n/2\leq k\leq n-2.$
We then
construct a new family of cross-bifix-free codes $U^{(t)}_{I,J}(n)$ to expand $S_{I,J}^{(k)}(n)$ such that
the resulting larger code $S_{I,J}^{(k)}(n)\bigcup U^{(t)}_{I,J}(n)$ is a non-expandable
cross-bifix-free code whenever $S_{I,J}^{(k)}(n)$ is expandable.
Finally, we present an explicit  formula for the size of
$S_{I,J}^{(k)}(n)\bigcup U^{(t)}_{I,J}(n)$.

There may be more than one way to expand  $S_{I,J}^{(k)}(n)$ to a non-expandable cross-bifix-free code
in the case where  $n/2\leq k\leq n-2.$
A possible
direction for future  work is to find  a new cross-bifix-free code  $U$ such that  $U\bigcup S_{I,J}^{(k)}(n)$
has the largest code size among all the expansion of  $S_{I,J}^{(k)}(n)$.


\end{document}